\documentclass[pre,aps]{revtex4}
\usepackage{graphicx,epsf}
\usepackage{amssymb}
\usepackage{amsmath}
\usepackage{amsthm}
\usepackage{latexsym}

\newtheorem {theorem} {Theorem}%[section]
\newtheorem {proposition} [theorem]{Proposition}

\def\be{\begin{equation}}
\def\ee{\end{equation}}
\def\bq{\begin{eqnarray}}
\def\eq{\end{eqnarray}}
\def\beq{\begin{eqnarray*}}
\def\eeq{\end{eqnarray*}}

\def \e {\varepsilon}

\begin{document}

\title{SLOW INVARIANT MANIFOLDS OF\\ SLOW-FAST DYNAMICAL SYSTEMS}

\author{Jean-Marc GINOUX}

\affiliation{Aix Marseille Univ, Universit\'{e} de Toulon, CNRS, CPT, Marseille, France, ginoux@univ-tln.fr}

\begin{abstract}
Slow-fast dynamical systems, i.e., singularly or non-singularly perturbed dynamical systems possess slow invariant manifolds on which trajectories evolve slowly. Since the last century various methods have been developed for approximating their equations. This paper aims, on the one hand, to propose a classification of the most important of them into two great categories: singular perturbation-based methods and curvature-based methods, and on the other hand, to prove the equivalence between any methods belonging to the same category and between the two categories. Then, a deep analysis and comparison between each of these methods enable to state the efficiency of the \textit{Flow Curvature Method} which is exemplified with paradigmatic Van der Pol singularly perturbed dynamical system and Lorenz slow-fast dynamical system.
\end{abstract}

\maketitle

\section{Introduction}
\noindent

A great number of phenomena can generally be modeled with dynamical systems, i.e., with sets of nonlinear ordinary differential equations. Among their huge variety, the evolution of some of them are characterized by the existence of at least two time scales: a \textit{slow} time and a \textit{fast} time. The metaphorical example of the Tantalus cup is often used to emphasize this kind of evolution since it fills slowly and quickly empties \cite{Lecorbeiller}. Such \textit{slow-fast} evolution is transcribed by the presence of at least a small multiplicative parameter $\varepsilon$ in the velocity vector field of these dynamical systems which have been called since \textit{singularly perturbed dynamical systems}.

The classical geometric theory of differential equations was originally developed by Poincar\'{e} in his famous memoirs entitled ``On the curves defined by differential equations'' \cite{P1,P2,P3,P4} in which he formalized the concept of \textit{dynamical systems} consisting of set of ordinary differential equations (ODE) as well as that of \textit{phase space} (For a History of Nonlinear Oscillations Theory, see Ginoux \cite{Ginoux2017}). Few years later, \textit{perturbation method}, which had been introduced by Sim\'{e}on Denis Poisson in 1830, were greatly improved by Poincar\'{e} in his \textit{New Methods of Celestial Mechanics} \cite{P5}. According to Minorsky \cite[p. 213]{Minorsky1962}:\\

\begin{quote}
``Poincar\'{e} opens the new period in the perturbation theory in that the periodicity begins to play the primary role, and this question became of fundamental importance in the theory of oscillations.''
\end{quote}

Indeed, this idea was emphasized in the very beginning of the chapter III of the first volume of his ``New Methods'' by Poincar\'{e} \cite[p. 82]{P5} through this famous sentence:\\

\begin{quote}
``In addition, these periodic solutions are so valuable for us because they are, so to say, the only breach by which we may attempt to enter an area heretofore deemed inaccessible.''
\end{quote}

In the 1930-1950s Andronov \& Chaikin \cite{Andro1}, Levinson \cite{Lev} and Tikhonov \cite{Tikh} generalized Poincar\'{e}'s ideas and stated that \textit{singularly perturbed systems} possess \textit{invariant manifolds} on which trajectories evolve slowly, and toward which nearby orbits contract exponentially in time (either forward or backward) in the normal directions. These manifolds have been called asymptotically stable (or unstable)
\textit{slow invariant manifolds}. Then, Fenichel theory \cite{Fen5, Fen6, Fen7, Fen8} for the \textit{persistence of normally hyperbolic invariant manifolds} enabled to establish the \textit{local invariance} of \textit{slow invariant manifolds} that possess both expanding and contracting directions and which were labeled \textit{slow invariant manifolds}. The theory of invariant manifolds for an ordinary  differential equation is based on the work of Hirsch, \textit{et al.} \cite{Hirsch}.

During the last century, various methods have been developed to compute the \textit{slow invariant manifold} or, at least an asymptotic expansion in power of $\varepsilon$.

The seminal works of Wasow \cite{Wasow}, Cole \cite{Cole}, O'Malley \cite{Malley1, Malley2} and Fenichel \cite{Fen5, Fen6, Fen7, Fen8} to name but a few, gave rise to the so-called \textit{Geometric Singular Perturbation Theory}. According to this theory, existence as well as local invariance of the \textit{slow invariant manifold} of \textit{singularly perturbed dynamical systems} has been stated. Then, the determination of the \textit{slow invariant manifold} equation turned into a regular perturbation problem in which one generally expected the asymptotic validity of such expansion to breakdown \cite{Malley2}. In the framework of the \textit{Geometric Singular Perturbation Theory}, the zero-order approximation in $\varepsilon$ of the \textit{slow invariant manifold} is called \textit{singular approximation} \cite{Rossetto1987} or \textit{critical manifold} \cite{Guckenheimer2000}. At the end of the 1980s, Rossetto \cite{Rossetto1986, Rossetto1987} developed the \textit{Successive Approximations Method} to approximate the \textit{slow invariant manifold} of \textit{singularly perturbed dynamical systems}. Then, Gear \textit{et al.} \cite{Gear} and Zagaris \textit{et al.} \cite{Zagaris} used the \textit{Zero-Derivative Principle} for the same purpose.

Beside, \textit{singularly perturbed dynamical systems} there are \textit{dynamical systems} without an explicit timescale splitting. In this case the \textit{Geometric Singular Perturbation Theory} can not be applied anymore. Some of these systems, which have been called \textit{slow-fast dynamical systems}, have the following property: their Jacobian matrix has at least a \textit{fast eigenvalue}, i.e. with the largest absolute value of the real part. Let's notice that a \textit{singularly perturbed dynamical system} is \textit{slow-fast} while a \textit{slow-fast dynamical system} is not necessary \textit{singularly perturbed} (For a proof of this statement see Rossetto \textit{et al.} \cite{Rossetto1998}). Thus, in the beginning of the 1990s various approaches have been proposed in order to approximate slow manifolds of such \textit{slow-fast dynamical systems}. As recalled by Lebiedz and Unger \cite{Lebiedz} ``In 1992, Maas and Pope introduced the intrinsic low-dimensional manifold (ILDM) method, which has become very popular and widely used in the reactive flow community, in particular in combustion applications.'' Two years later, Br{\o}ns and Bar-Eli \cite{Brons1994} presented an approach, only applicable to two-dimensional systems, based on the \textit{inflection line}, i.e. the locus of points where the curvature of trajectories vanishes, to provide an approximation of the slow manifold. Then, Rossetto \textit{et al.} \cite{Rossetto1998} proposed to use the \textit{Tangent Linear System Approximation} for the same purpose. Since 2005, a new approach of $n$-dimensional \textit{singularly perturbed dynamical systems} or \textit{slow-fast dynamical systems} based on the location of the points where the \textit{curvature} of \textit{trajectory curves} vanishes, called \textit{Flow Curvature Method} has been developed by Ginoux \textit{et al.} \cite{GiRo1,GiRo2} and Ginoux \cite{Gin}. So, the aim of this paper is to propose a classification of the most important method of approximation of slow manifolds into two great categories:

\smallskip

\begin{itemize}
\item \textit{singular perturbation-based methods} and
\item \textit{curvature-based methods}.
\end{itemize}

\smallskip

Thus, it is stated that \textit{Geometric Singular Perturbation Method} encompasses \textit{Successive Approximations Method} and \textit{Zero-Derivative Principle} which are perfectly identical and belong to the first category. Then, we show on the one hand that both \textit{Intrinsic Low-Dimensional Manifold Method} (IDLM) and \textit{Tangent Linear System Approximation} (TLSA) are also identical. On the other hand we prove that the \textit{Flow Curvature Method} is a generalization of the \textit{Inflection Line Method} (ILM) which cannot in anyway be extended to higher dimensions than two. The reasons of such impossibility which mainly consists in the arc length parametrization are fully detailed. Finally, we recall the identity between \textit{Flow Curvature Method} and \textit{Tangent Linear System Approximation} established in our previous works more than ten years ago. It follows that \textit{Flow Curvature Method} encompasses these three other methods (IDLM, TLSA, ILM) which belong to the second category. After having divided these methods of approximation of slow manifolds into two categories, we also recall the identity between \textit{Flow Curvature Method} and \textit{Geometric Singular Perturbation Method} for $n$-dimensional \textit{singularly perturbed dynamical systems} up to suitable order in $\varepsilon$. At last, a deep analysis and comparison between each of these various methods enable to state the efficiency of the \textit{Flow Curvature Method} what is exemplified with paradigmatic \textit{singularly perturbed dynamical systems} and \textit{slow-fast dynamical systems} of dimension two and three.

\section{Singular Perturbation-Based Methods}

\subsection{Singularly perturbed dynamical systems}

Following the works of Jones \cite{Jones} and Kaper \cite{Kaper} some fundamental concepts and definitions for systems of ordinary differential equations with two time scales, i.e., for \textit{singularly perturbed dynamical systems} are briefly recalled. Thus, we consider systems of differential equations of the form:

\begin{equation}
\label{eq1}
\left\{
\begin{aligned}
 {\vec {x}}' & = \vec {f}\left( {\vec {x},\vec {y},\varepsilon } \right), \hfill \\
 {\vec {y}}' & = \varepsilon \vec {g}\left( {\vec {x},\vec {y},\varepsilon } \right), \hfill \\
\end{aligned}
\right.
\end{equation}

where $\vec {x} \in \mathbb{R}^m$, $\vec {y} \in \mathbb{R}^p$, $\varepsilon \in \mathbb{R}^ + $, and the prime denotes differentiation with respect to the independent variable $t$. The functions $\vec {f}$ and $\vec {g}$ are assumed to be $C^\infty$ functions (In certain applications these functions will be
supposed to be $C^r$, $r \geqslant 1$) of $\vec {x}$, $\vec {y}$ and $\varepsilon$ in $U\times I$, where $U$ is an open subset of $\mathbb{R}^m\times \mathbb{R}^p$ and $I$ is an open interval containing $\varepsilon = 0$.

In the case when $\varepsilon \ll 1$, i.e., is a small positive number, the variable $\vec {x}$ is called \textit{fast} variable, and $\vec {y}$ is called \textit{slow} variable. Using Landau's notation: $O ( \varepsilon ^l )$ represents a real polynomial in $\varepsilon $ of $l$ degree, with $l \in \mathbb{Z}$,
it is used to consider that generally $\vec {x}$ evolves at an $O\left( 1 \right)$ rate; while $\vec {y}$ evolves at an $O\left( \varepsilon \right)$ \textit{slow} rate.

Reformulating the system (\ref{eq1}) in terms of the rescaled variable $\tau = \varepsilon t$, we obtain:

\begin{equation}
\label{eq2}
\left\{
\begin{aligned}
 \varepsilon \dot {\vec {x}} & = \vec {f}\left( {\vec {x},\vec{y},\varepsilon } \right), \hfill \\
 \dot {\vec {y}} & = \vec {g}\left( {\vec {x},\vec {y},\varepsilon } \right). \hfill \\
\end{aligned}
\right.
\end{equation}

The dot $\left( \cdot \right)$ represents the derivative with respect to the new independent variable $\tau $. The independent variables $t$ and $\tau $ are referred to the \textit{fast} and \textit{slow} times, respectively, and (\ref{eq1}) and (\ref{eq2}) are called the \textit{fast} and \textit{slow} systems, respectively. These systems are equivalent whenever $\varepsilon \ne 0$, and they are labeled \textit{singular perturbation problems} when $\varepsilon \ll 1$, i.e., is a small positive parameter. The label ``singular'' stems in part from the discontinuous limiting behavior in the system (\ref{eq1}) as $\varepsilon \to 0$. In such case, the system (\ref{eq1}) reduces to an $m$-dimensional system called \textit{reduced fast system}, with the variable $\vec{y}$ as a constant parameter:

\begin{equation}
\label{eq3}
\left\{
\begin{aligned}
{\vec {x}}' & = \vec {f}\left( {\vec {x}, \vec {y}, 0 } \right), \hfill \\
{\vec {y}}' & = \vec {0}. \hfill \\
\end{aligned}
\right.
\end{equation}

System (\ref{eq2}) leads to the following differential-algebraic system called \textit{reduced slow system} which dimension decreases from $m + p$ to $p$:

\begin{equation}
\label{eq4}
\left\{
\begin{aligned}
\vec {0} & = \vec {f}\left( \vec {x},\vec {y},0  \right), \hfill\\
\dot {\vec {y}} & = \vec {g} \left( \vec {x},\vec {y}, 0  \right).\hfill \\
\end{aligned}
\right.
\end{equation}

By exploiting the decomposition into \textit{fast} and \textit{slow} reduced systems (\ref{eq3}) and (\ref{eq4}), the geometric approach reduced the full \textit{singularly perturbed system} to separate lower-dimensional regular perturbation problems in the \textit{fast} and \textit{slow} regimes, respectively.

\subsection{Fenichel geometric theory}

Fenichel geometric theory for general systems (\ref{eq1}), i.e., a theorem providing
conditions under which \textit{normally hyperbolic invariant manifolds} in system (\ref{eq1}) persist when the perturbation is
turned on, i.e., when $0 < \varepsilon \ll 1$ is briefly recalled in this subsection. This theorem concerns only compact manifolds with boundary.

\subsubsection{Normally hyperbolic manifolds}

Let's make the following assumptions about system (\ref{eq1}):\\

\textbf{(H}$_{1}$\textbf{)} \textit{The functions }$\vec
{f}$ \textit{and} $\vec {g}$ \textit{are} $C^\infty $\textit{ in
}$U\times I$\textit{, where U is an open subset of
}$\mathbb{R}^m\times \mathbb{R}^p$\textit{ and I is an open interval
containing }$\varepsilon = 0.$

\textbf{(H}$_{2}$\textbf{)} \textit{There exists a set }$M_0
$\textit{ that is contained in }$\left\{ {\left( {\vec {x},\vec {y}}
\right):\vec {f}\left( {\vec {x},\vec {y},0} \right) = {\vec {0}}}
\right\}$\textit{ such that }$M_0 $\textit{ is a compact manifold
with boundary and }$M_0 $\textit{ is given by the graph of a
}$C^\infty $\textit{ function }$\vec {y} = \vec {Y}_0 \left( \vec
{x} \right)$\textit{ for }$\vec {x} \in D$\textit{, where }$D
\subseteq \mathbb{R}^p$\textit{ is a compact, simply connected
domain and the boundary of D is an} $\left( {p - 1} \right)$
\textit{dimensional} $C^\infty$ \textit{ submanifold. Finally, the
set D is overflowing invariant with respect to (\ref{eq2}) when
}$\varepsilon = 0. $

\textbf{(H}$_{3}$\textbf{)} $M_0 $\textit{ is normally hyperbolic
relative to (\ref{eq3}) and in particular it is required for all
points }$\vec {p} \in M_0 $\textit{, that there are k (resp. l)
eigenvalues of }$D_{\vec{y}} \vec{f}\left( {\vec{p},0}
\right)$\textit{ with positive (resp. negative) real parts bounded
away from zero, where }$k + l = m.$

\subsubsection{Fenichel persistence theory for singularly perturbed systems}

For compact manifolds with boundary, Fenichel's persistence theory states that, provided the hypotheses $(H_{1})-(H_{3})$ are satisfied, the system (\ref{eq1}) has a \textit{slow} (or center) \textit{manifold}, and this \textit{slow manifold} has \textit{fast} stable and unstable \textit{manifolds}.\\

\newpage

\textbf{\textit{Theorem for compact manifolds with boundary:}}

\textit{Let system (\ref{eq1}) satisfy the conditions (H}$_{1})-(H_{3}$\textit{). If }$\varepsilon > 0$\textit{ is sufficiently small, then there exists a function }$\vec {Y}\left( {\vec {x},\varepsilon } \right)$\textit{ defined on D such that the manifold }$M_\varepsilon = \{ {\left( {\vec {x},\vec {y}} \right):\vec {y} = \vec {Y}\left( {\vec {x},\varepsilon } \right)} \}$ \textit{is locally invariant under (\ref{eq1}). Moreover, }$\vec {Y}\left( {\vec {x},\varepsilon } \right)$\textit{ is }$C^r$\textit{ for any }$r < + \infty $, \textit{ and }$M_\varepsilon $\textit{ is }$C^rO\left( \varepsilon \right)$
\textit{close to }$M_0 $\textit{. In addition, there exist perturbed local stable and unstable manifolds of }$M_\varepsilon $\textit{. They are unions of invariant families of stable and unstable fibers of dimensions l and k, respectively, and they are} $C^rO\left( \varepsilon \right)$ \textit{close for all} $r < + \infty$ \textit{, to their counterparts.}\\

\begin{proof}
For proof of this theorem see Fenichel \cite{Fen5,Fen6,Fen7,Fen8}.
\end{proof}

The label slow manifold is attached to $M_\varepsilon $ because the magnitude of the vector field restricted to $M_\varepsilon $ is $O\left( \varepsilon \right)$, in terms of the fast independent variable $t$. So persistent manifolds are labeled \textit{slow manifolds}, and the proof of their persistence is carried out by demonstrating that the local stable and unstable manifolds of $M_0$ also persist as locally invariant manifolds in the perturbed system, i.e., that the local hyperbolic structure persists, and then the slow manifold is immediately at hand as a locally invariant manifold in the transverse intersection of these persistent local stable and unstable manifolds.

\subsection{Geometric Singular Perturbation Method}

Earliest geometric approaches to \textit{singularly perturbed systems} have been developed by Cole \cite{Cole}, O'Malley \cite{Malley1, Malley2}, Fenichel \cite{Fen5,Fen6,Fen7,Fen8} for the determination of the \textit{slow manifold} equation. Generally, Fenichel theory enables to turn the problem for explicitly finding functions $\vec {y} = \vec {Y}\left( {\vec {x},\varepsilon } \right)$ whose graphs are locally invariant \textit{slow manifolds} $M_\varepsilon $ of system (\ref{eq1}) into regular perturbation problem \cite[p. 112]{Malley1}. Invariance of the manifold $M_\varepsilon $ implies that $\vec {Y}\left( {\vec {x},\varepsilon } \right)$ satisfies:

\begin{equation}
\label{eq5}
D_{\vec {x}} \vec {Y}\left( {\vec
{x},\varepsilon } \right)\vec {f}(\vec{x}, {\vec {Y}\left( {\vec
{x},\varepsilon } \right),\varepsilon } ) = \varepsilon  \vec
{g}(\vec{x}, {\vec {Y}\left( {\vec {x},\varepsilon } \right),\varepsilon } ).
\end{equation}

According to Guckenheimer \textit{et al.} \cite[p.131]{Guck1983}, this (partial)
differential equation for $\vec {Y}\left( {\vec {x}, \varepsilon} \right)$ cannot be
solved exactly. So, its solution can be approximated arbitrarily closely as
a Taylor series at $\left( {\vec {x}, \varepsilon} \right) = ( \vec {0}, 0 )$. Then, the following perturbation expansion is plugged:

\begin{equation}
\label{eq6}
\vec{Y}\left( {\vec {x},\varepsilon } \right) = \vec {Y}_0 \left( \vec
{x} \right) + \varepsilon \vec {Y}_1 \left( \vec {x} \right) + O\left( {\varepsilon^2} \right)
\end{equation}

into (\ref{eq5}) to solve order by order for $\vec {Y}\left( {\vec {x},\varepsilon } \right)$. The Taylor series expansion \cite{Malley2} for $\vec {f}(\vec{x}, {\vec {Y}\left( {\vec {x},\varepsilon } \right),\varepsilon } )$ and $\vec {g}(\vec{x}, {\vec {Y}\left({\vec {x},\varepsilon } \right),\varepsilon })$ up to terms of order two in $\varepsilon $ reads:

\[
\begin{aligned}
& \vec {f}(\vec{x}, {\vec {Y}\left( {\vec {x},\varepsilon } \right),\varepsilon })  =  \vec {f}(\vec{x}, {\vec {Y}_0 \left( \vec {x} \right),0}) +  \varepsilon \left[ D_{\vec {y}} \vec {f}(\vec{x}, {\vec {Y}_0 \left( \vec {x} \right),0} ) \vec {Y}_1\left( \vec {x} \right) + \frac{\partial \vec {f}}{\partial \varepsilon } (\vec{x}, {\vec {Y}_0 \left( \vec {x} \right),0} ) \right], \\
& \vec {g}(\vec{x}, {\vec {Y}\left( {\vec {x},\varepsilon } \right),\varepsilon } )  =  \vec {g} (\vec{x}, {\vec {Y}_0 \left( \vec {x} \right),0} ) +  \varepsilon \left[ D_{\vec {y}} \vec {g}(\vec{x}, {\vec {Y}_0 \left( \vec {x} \right),0} ) \vec {Y}_1\left( \vec {x} \right) + \frac{\partial \vec {g}}{\partial \varepsilon }(\vec{x}, {\vec {Y}_0 \left( \vec {x} \right),0} ) \right].
\end{aligned}
\]

\textbullet \hspace{0.1in} At order $\varepsilon ^0$, Eq. (\ref{eq5}) gives:

\begin{equation}
\label{eq7}
D_{\vec {x}} \vec{Y_0}\left( \vec{x}\right)  \vec {f}(\vec{x}, {\vec {Y}_0 \left( \vec {x} \right),0}) = \vec {0},
\end{equation}

which defines $\vec {Y}_0 \left( \vec {x} \right)$ due to the invertibility of $D_{\vec {y}} \vec {f}$ and the \textit{Implicit Function Theorem}.\\

\textbullet \hspace{0.1in} The next order $\varepsilon ^1$ provides:

\begin{equation}
\label{eq8}
D_{\vec {x}} \vec {Y}_0 \left( \vec {x} \right) \left[ D_{\vec {y}} \vec {f} (\vec{x}, {\vec {Y}_0 \left( \vec {x} \right),0} ) \vec{Y_1}\left( \vec{x} \right) + \frac{\partial \vec {f}}{\partial \varepsilon } \right] = \vec {g}(\vec{x}, {\vec {Y}_0 \left( \vec {x} \right),0} ),
\end{equation}

which yields $\vec {Y}_1 \left( \vec {x} \right)$ and so forth.

\newpage

So, regular perturbation theory makes it possible to build an approximation of locally invariant \textit{slow manifolds} $M_\varepsilon$. Thus, in the framework of the \textit{Geometric Singular Perturbation Method}, three conditions are needed to characterize the \textit{slow manifold} associated with \textit{singularly perturbed system}: existence, local invariance and determination. Existence and local invariance of the \textit{slow manifold} are stated according to Fenichel theorem for compact manifolds with boundary while \textit{asymptotic expansions} provide its equation up to the order of the expansion.

\subsubsection{Van der Pol system}

The Van der Pol system \cite{Vdp1926} was introduced in the middle of the twenties for modelling \textit{relaxation oscillations}. This model presented below is nowadays considered as a paradigm of \textit{singularly perturbed dynamical systems}.

\begin{equation}
\label{eq9}
\left\{
\begin{aligned}
x' & = f\left( x, y, \e \right) = x + y - \dfrac{x^3}{3}, \hfill \\
y' & = g\left( x, y, \e \right) = - \varepsilon x. \hfill \\
\end{aligned}
\right.
\end{equation}

According to Eq. (\ref{eq5}) invariance of the manifold $M_{\varepsilon}$ reads:

\begin{equation}
\label{eq10}
\left( \dfrac{\partial Y}{\partial x}(x,\e) \right) f\left( x, Y(x,\e),\e \right) = \varepsilon g\left(x, Y(x,\e),\e \right).
\end{equation}

By application of the \textit{Geometric Singular Perturbation Method} and by solving equation (\ref{eq10}) order by order provides at:\\

\textbf{Order $\varepsilon^0$}

\begin{equation}
\label{eq11}
\dfrac{\partial Y_0}{\partial x}(x) f(x, Y_0(x),0) = 0 \quad \Leftrightarrow \quad  f(x, Y_0(x),0) = 0 \quad \Leftrightarrow \quad Y_0(x) = \dfrac{x^3}{3} - x.
\end{equation}

\textbf{Order $\varepsilon^1$}

\begin{equation}
\label{eq12}
\dfrac{\partial Y_0}{\partial x}(x) Y_1(x) =  g(x, Y_0(x),0) \quad \Leftrightarrow \quad  Y_1(x) =  \dfrac{x}{1-x^2}.
\end{equation}

Thus, according to Eq. (\ref{eq6}), and following this method an approximation up to order $\e^3$ of the \textit{slow invariant manifold} of the Van der Pol \textit{singularly perturbed dynamical system} is given by:

\begin{equation}
\label{eq13}
y = Y(x, \e) = \dfrac{x^3}{3} - x + \e \dfrac{x}{1-x^2}  + \e^2 \dfrac{x (1 + x^2)}{(1 - x^2)^4} + O\left( {\varepsilon^3} \right).
\end{equation}

\subsection{Successive Approximations Method}

At the end of the 1980s, Rossetto \cite{Rossetto1986, Rossetto1987} proposed the \textit{Successive Approximations Method} to approximate the \textit{slow invariant manifold} of \textit{singularly perturbed dynamical systems}. He considered that the \textit{singular approximation}, i.e., the zero-order approximation in $\varepsilon$ of the {\it slow invariant manifold}, i.e., the \textit{critical manifold} $M_0$ defined by $\vec {f}\left( {\vec {x},\vec {y},\varepsilon } \right) = \vec{0}$ is overflowing invariant with respect to (\ref{eq1}) when $\varepsilon = 0$. Invariance of the manifold $M_0$ implies that $\vec {f}\left( {\vec {x},\vec {y},\varepsilon } \right)$ satisfies:

\begin{equation}
\label{eq14}
\dfrac{d \vec{f}}{dt}\left( {\vec {x},\vec {y},\varepsilon } \right) = \vec{0} \quad \Leftrightarrow \quad (D_{\vec {x}} \vec{f} )\vec{x}' + (D_{\vec {y}} \vec{f})\vec{y}' = \vec{0}.
\end{equation}

According to (\ref{eq1}), this leads to:

\begin{equation}
\label{eq15}
(D_{\vec {x}} \vec{f} )\vec{f} ( {\vec {x},\vec {y},\varepsilon } ) + \varepsilon (D_{\vec {y}} \vec{f} ) \vec{g} ( {\vec {x},\vec {y},\varepsilon } ) = \vec{0}.
\end{equation}

By application of the \textit{Implicit Function Theorem}, we have:

\begin{equation}
\label{eq16}
\vec{f} ( \vec {x},\vec {y},\varepsilon) = \vec{0} \quad \Leftrightarrow \quad \vec{y} = \vec{Y}\left( {\vec {x},\varepsilon } \right),
\end{equation}

implies that:

\begin{equation}
\label{eq17}
D_{\vec {x}} \vec{Y} = - (D_{\vec {y}} \vec{f})^{-1}(D_{\vec {x}} \vec{f}).
\end{equation}

Left multiplication of (\ref{eq15}) by $(D_{\vec {y}} \vec{f})^{-1}$ gives:

\begin{equation}
\label{eq18}
(D_{\vec {y}} \vec{f})^{-1}(D_{\vec {x}} \vec{f} )\vec{f} ( {\vec {x},\vec {y},\varepsilon } ) + \varepsilon \vec{g} ( {\vec {x},\vec {y},\varepsilon } ) = \vec{0}.
\end{equation}

Taking into account (\ref{eq16}) \& (\ref{eq17}), we find the following equation:

\begin{equation}
\label{eq19}
(D_{\vec {x}} \vec{Y} ) \vec{f} ( {\vec {x},\vec {Y}(\vec{x}, \e),\varepsilon } ) = \varepsilon \vec{g} ( {\vec {x},\vec {Y}(\vec{x}, \e),\varepsilon } ),
\end{equation}

which is absolutely identical to (\ref{eq5}). Thus, identity between \textit{Geometric Singular Perturbation Method} and \textit{Successive Approximations Method} is stated.

\subsection{Zero-Derivative Principle}

In 2005, an iterative model reduction method called \textit{Zero-Derivative Principle} (ZDP) was presented by Gear \textit{et al.} \cite{Gear}. Four years later, Zagaris \textit{et al.} \cite{Zagaris} applied the ZDP to \textit{singularly perturbed dynamical systems} (\ref{eq1}) to approximate their \textit{slow invariant manifold}. They considered that if a function is locally invariant under (\ref{eq1}) the zero set of its time derivative provides an even more accurate approximation of the \textit{slow invariant manifold} (SIM) of (\ref{eq1}). In other words, it means that the $k^{th}$ approximation of the \textit{slow invariant manifold} of (\ref{eq1}) is given by the following equation:

\begin{equation}
\label{eq20}
\dfrac{d^{k+1}\vec{x}}{dt^{k+1}} = \dfrac{d^k\vec{f}}{dt^k}\left( {\vec {x},\vec {y},\varepsilon } \right) = \vec{0}.
\end{equation}

According to Beno\^{i}t \textit{et al.} \cite{Benoit2015}: ``estimates [of the SIM] can be obtained in terms of approximations of the slow manifold to order $\e^k$ when looking at the zeros of the $k^{th}$ derivative of the fast components of the original singularly perturbed system.'' First of all, let's notice that such result has been already stated and published by Rossetto \cite{Rossetto1986, Rossetto1987} and has not been quoted in any reference. In his works, Rossetto \cite{Rossetto1986, Rossetto1987} presented his \textit{Successive Approximations Method} and stated that:

\begin{quote}
``The $m^{th}$ order approximation of slow motion is solution of the dynamical system $\vec{f}^{(m)} = \vec{0}$.''
\end{quote}

Thus, according to the ZDP, if $\vec{x}' = \vec {f}\left( {\vec {x},\vec {y},\varepsilon } \right) = \vec{0}$ provides the zero-order approximation in $\varepsilon$ of the \textit{slow invariant manifold} of (\ref{eq1}), $\vec{x}'' = \vec{0}$ will define its first-order approximation in $\varepsilon$. So, according to (\ref{eq1}), we have the following equation:

\begin{equation}
\label{eq21}
\vec{x}''= (D_{\vec {x}} \vec{f}) \vec{x}' + (D_{\vec {y}} \vec{f}) \vec{y}' = \vec{0},
\end{equation}

which reads:

\begin{equation}
\label{eq22}
(D_{\vec {x}} \vec{f}) \vec {f}\left( {\vec {x},\vec {y},\varepsilon } \right)  + \varepsilon (D_{\vec {y}} \vec{f}) \vec {g}\left( {\vec {x},\vec {y},\varepsilon } \right)  = \vec{0}.
\end{equation}

This equation (\ref{eq22}) is absolutely identical to Eq. (\ref{eq15}). This proves that both \textit{Successive Approximations Method} and \textit{Zero-Derivative Principle} are identical. Moreover, since we have already stated above that both \textit{Geometric Singular Perturbation Method} and \textit{Successive Approximations Method} are identical, it follows by transitivity that these three methods are absolutely identical and so belong to the first category called: \textit{Singular Perturbation-Based Methods}.

\section{Curvature-Based Methods}

As previously recalled, beside, \textit{singularly perturbed dynamical systems} (\ref{eq1}) there are \textit{dynamical systems} without an explicit timescale splitting. In this case the \textit{Singular Perturbation-Based Methods}, i.e., GSPM, SAM and ZDP can not be applied anymore. Some of these systems, called \textit{slow-fast dynamical systems}, may be defined by the following set of nonlinear ordinary differential equations:

\begin{equation}
\label{eq23}
\frac{d\vec {X}}{dt} = \overrightarrow  \Im ( \vec {X} ),
\end{equation}

\smallskip

with $\vec {X} = \left[ {x_1 ,x_2 ,...,x_n } \right]^t \in E \subset {\mathbb{R}}^n$ and $\overrightarrow \Im ( \vec {X} ) = \left[ {f_1 ( \vec {X}),f_2 ( \vec {X} ),...,f_n ( \vec {X} )} \right]^t \in E \subset {\mathbb{R}}^n$. The vector $\overrightarrow \Im ( \vec {X} )$ defines a velocity vector field $\dot{\vec{X}} = \vec V$ in $E$ whose components $f_i $ which are supposed to be continuous and infinitely differentiable with respect to all $x_i $ and $t$, i.e. are $C^\infty $ functions in $E$ and with values included in $\mathbb{R}$, satisfy the assumptions of the Cauchy-Lipschitz theorem. For more details, see for example Coddington \& Levinson \cite{cod}. A solution of this system is a \textit{trajectory curve} $\vec {X}\left( t \right)$ tangent (Except at the \textit{fixed points}) to $\overrightarrow \Im$ whose values define the \textit{states} of the \textit{dynamical system} described by the Eq. (\ref{eq23}). According to Ginoux \cite{Gin}, \textit{non-singularly perturbed dynamical systems} (\ref{eq23}) can be considered as \textit{slow-fast dynamical systems} if their Jacobian matrix has at least one \textit{fast eigenvalue}, i.e. with the largest absolute value of the real part over a huge domain of the phase space.

\subsection{Intrinsic Low-Dimensional Manifold Method}

In 1992, Maas and Pope introduced the \textit{Intrinsic Low-Dimensional Manifold} (ILDM) method \cite{Maas}. Within the framework of application of the Tikhonov's theorem \cite{Tikh}, this method uses the fact that in the vicinity of the \textit{slow manifold} the eigenmode associated with the \textit{fast} eigenvalue is evanescent. According to Mass and Pope \cite{Maas}:

\begin{quote}
``If a subspace in the composition space can be found where the system is ``in equilibrium'' with respect to its smallest $m$ eigenvectors, this subspace defines a low-dimensional manifold that is characterized by the fact that movements along it are associated with slow time scales and can be used to simplify chemical kinetics.''
\end{quote}

Thus, IDLM method provides an approximation of the \textit{slow invariant manifold} of the \textit{slow-fast dynamical system} (\ref{eq23}) by considering that its velocity vector field is coplanar to the \textit{slow} eigenvectors of the Jacobian of $\overrightarrow  \Im ( \vec {X} )$.

\subsection{Tangent Linear System Approximation Method}

At the end of the nineties, Rossetto \textit{et al.} \cite{Rossetto1998} presented their \textit{Tangent Linear System Approximation Method}. The application of this method requires that the \textit{slow-fast dynamical system} (\ref{eq23}) satisfies the following assumptions:

\begin{itemize}
\item[(H$_{1}$)] The components $f_i $, of the velocity vector field $\overrightarrow \Im ( \vec {X} )$ defined in $E$ are continuous, $C^\infty $ functions in $E$ and with values included in $\mathbb{R}$.
\item[(H$_{2}$)] The \textit{slow-fast dynamical system} (\ref{eq23}) satisfies the \textit{nonlinear part condition}, i.e., that the influence of the nonlinear part of the Taylor series of the velocity vector field $\overrightarrow \Im ( \vec {X} )$ of this system is overshadowed by the fast dynamics of the linear part.
\end{itemize}

\begin{equation}
\label{eq24}
\overrightarrow \Im ( \vec {X} ) = \overrightarrow \Im ( \vec{X}_0  ) + ( \vec{X} - \vec{X}_0 ) \left.  D_{\vec {X}} \overrightarrow \Im ( \vec {X} ) \right|_{\vec {X}_0}  + O\left( ( \vec {X} - \vec {X}_0  )^2 \right).
\end{equation}

To the \textit{slow-fast dynamical system} (\ref{eq23}) is associated a \textit{tangent linear system} defined as follows:

\begin{equation}
\label{eq25}
\frac{d\delta \vec {X}}{dt} = J( \vec {X}_0 )\delta \vec {X},
\end{equation}

where $\delta \vec {X} = \vec {X} - \vec {X}_0 $, $\vec{X}_0 = \vec{X}( t_0 )$ and $J(\vec {X}_0 ) = \left.  D_{\vec {X}} \overrightarrow \Im ( \vec {X} ) \right|_{\vec {X}_0}$ is the Jacobian matrix. Thus, the solution of the \textit{tangent linear system} (\ref{eq25}) reads:

\begin{equation}
\label{eq26}
\delta \vec {X} = e^{J\left( {\vec {X}_0 } \right)\left( {t - t_0 } \right)}\delta \vec {X}\left( {t_0 } \right).
\end{equation}

So,

\begin{equation}
\label{eq27}
\delta \vec {X} = \sum\limits_{i = 1}^n {a_i} \vec{Y}_{\lambda_i}.
\end{equation}

where $n$ is the dimension of the eigenspace, $a_{i}$ represents coefficients depending explicitly on the co-ordinates of space and implicitly on time and $\vec{Y}_{\lambda_i}$ the eigenvectors associated with the Jacobian matrix of the \textit{tangent linear system}. Thus, according to the TLSA method we have the following proposition:

\begin{proposition}
In the vicinity of the \textit{slow manifold} the velocity of the \textit{slow-fast dynamical system} (\ref{eq23}) and that of the \textit{tangent linear system} (\ref{eq25}) merge.

\begin{equation}
\label{eq28}
\frac{d\delta \vec {X}}{dt} = \vec{V} _T \approx \vec{V},
\end{equation}

where $\vec{V}_T$ represents the velocity vector associated with the \textit{tangent linear system}.
\end{proposition}

Hence, the \textit{Tangent Linear System Approximation} method consists in the projection of the velocity vector field $\vec{V}$ on the eigenbasis associated with the eigenvalues of the functional Jacobian of the \textit{tangent linear system}. Indeed, by taking account of (\ref{eq25}) and (\ref{eq27}) we have according to (\ref{eq28}):

\begin{equation}
\label{eq29}
\frac{d\delta \vec {X}}{dt} = J( \vec {X}_0 )\delta \vec {X} = J( \vec {X}_0 ) \sum\limits_{i = 1}^n {a_i} \vec{Y}_{\lambda_i} = \sum\limits_{i = 1}^n {a_i } J( \vec {X}_0 ) \vec{Y}_{\lambda _i } = \sum\limits_{i = 1}^n {a_i } \lambda _i \vec{Y}_{\lambda_i}.
\end{equation}

Thus, Proposition 1 leads to:

\begin{equation}
\label{eq30}
\frac{d\delta \vec {X}}{dt} = \vec{V}_T \approx \vec{V} = \sum\limits_{i = 1}^n {a_i } \lambda _i \vec{Y}_{\lambda _i}.
\end{equation}

The equation (\ref{eq30}) constitutes in dimension two (resp. dimension three) a condition called \textit{collinearity} (resp. \textit{coplanarity}) condition which provides the analytical equation of the \textit{slow manifold} of the \textit{slow-fast dynamical system} (\ref{eq23}). An alternative proposed by Rossetto \textit{et al.} \cite{Rossetto1998} uses the ``fast'' eigenvector on the left associated with the ``fast'' eigenvalue of the transposed functional Jacobian of the \textit{tangent linear system}. In this case the velocity vector field $\vec{V}$ is then orthogonal with the ``fast'' eigenvector on the left. This constitutes a condition called \textit{orthogonality} condition which provides the analytical equation of the \textit{slow manifold} of the \textit{slow-fast dynamical system} (\ref{eq23}) (See Rossetto \textit{et al.} \cite{Rossetto1998} and Ginoux \textit{et al.} \cite{GiRo1, GiRo2}).

For a two-dimensional \textit{slow-fast dynamical system} (\ref{eq23}) the projection of the velocity vector field $\vec V $ on the eigenbasis reads thus:

\[
\frac{d\delta \vec {X}}{dt} = \vec V _T \approx \vec V = \sum\limits_{i = 1}^2 {a_i } \lambda _i \vec{Y}_{\lambda _i } = \alpha \vec{Y}_{\lambda _1 } + \beta \vec{Y}_{\lambda _2 },
\]

where $\alpha $ and $\beta $ represent coefficients depending explicitly on the co-ordinates of space and implicitly on time and where $\vec{Y}_{\lambda _1 }$ represents the ``fast'' eigenvector and $\vec{Y}_{\lambda _2 }$ the ``slow'' eigenvector. The existence of an evanescent mode in the vicinity of the \textit{slow manifold} implies according to Tikhonov's theorem \cite{Tikh}: $\alpha \ll 1$. We deduce:

\[
\vec V \approx \beta \vec{Y}_{\lambda _2 }.
\]

So, the following \textit{collinearity condition} between the velocity vector field $\vec V$ of a two-dimensional \textit{slow-fast dynamical system} (\ref{eq23}) and the ``slow'' eigenvector $\vec{Y}_{\lambda _2 }$ provides an approximation of its \textit{slow manifold}:

\begin{equation}
\label{eq31}
\vec V \wedge \vec{Y}_{\lambda _2 } = \vec{0}.
\end{equation}

For a three-dimensional \textit{slow-fast dynamical system} (\ref{eq23}) the projection of the velocity vector field $\vec V$ on the eigenbasis reads thus:

\[
\frac{d\delta \vec {X}}{dt} = \vec V _T \approx \vec V = \sum\limits_{i = 1}^3 {a_i } \lambda _i \vec{Y}_{\lambda _i } = \alpha \vec{Y}_{\lambda _1 } + \beta \vec{Y}_{\lambda _2 } + \delta \vec{Y}_{\lambda _3},
\]

where $\alpha $, $\beta $ and $\delta$ represent coefficients depending explicitly on the co-ordinates of space and implicitly on time and where $\vec{Y}_{\lambda _1 }$ represents the ``fast'' eigenvector and $\vec{Y}_{\lambda _2 },  \vec{Y}_{\lambda _3 }$ the ``slow'' eigenvectors. The existence of an evanescent mode in the vicinity of the \textit{slow manifold} implies according to Tikhonov's theorem \cite{Tikh}: $\alpha \ll 1$. We deduce:

\[
\vec V \approx \beta \vec{Y}_{\lambda _2 } + \delta \vec{Y}_{\lambda _3}.
\]

So, the following \textit{coplanarity condition} between the velocity vector field $\vec V$ of a three-dimensional \textit{slow-fast dynamical system} (\ref{eq23}) and the ``slow'' eigenvectors $\vec{Y}_{\lambda _2 }$ and $\vec{Y}_{\lambda _3}$ provides an approximation of its \textit{slow manifold}:

\begin{equation}
\label{eq32}
\vec V .( \vec{Y}_{\lambda _2 } \wedge \vec{Y}_{\lambda _3 }  ) = 0.
\end{equation}

\smallskip

This brief presentation enables to state that both \textit{Intrinsic Low-Dimensional Manifold} method and \textit{Tangent Linear System Approximation} method are absolutely identical. Let's notice that these methods can be also applied to \textit{singularly perturbed dynamical systems} (\ref{eq1}). However, they presented a major drawback since they required the computation of eigenvalues and eigenvectors of the Jacobian which could only be done numerically for \textit{dynamical systems} of dimension greater than two. Moreover, according to the nature of the ``slow'' eigenvalues (real or complex conjugated) the plot of their \textit{slow manifold} analytical equation was difficult even impossible. So, to solve this problem it was necessary to make the \textit{slow manifold} analytical equation independent of the ``slow'' eigenvalues. That's was the aim of the \textit{Flow Curvature Method}.

\subsection{Flow Curvature Method}

Fifteen years ago, Ginoux \textit{et al.} \cite{GiRo1} published the first article in which the \textit{Flow Curvature Method} was presented. They explained that by multiplying the \textit{slow manifold} analytical equation of a two dimensional dynamical system by a ``conjugated'' equation, that of a three dimensional dynamical system by two ``conjugated'' equations make the \textit{slow manifold} analytical equation independent of the ``slow'' eigenvalues of the
\textit{tangent linear system}.

For a two-dimensional \textit{slow-fast dynamical system} (\ref{eq23}), let's multiply equation (\ref{eq31}) by its ``conjugated'' equation, i.e., by an equation in which the eigenvalue $\lambda_2$ is replaced by the eigenvalue $\lambda_1$. Let's notice that the ``conjugated'' equation of the equation (\ref{eq31}) corresponds to the collinearity condition between the velocity vector field $\vec V $ and the fast eigenvector $\vec{Y}_{\lambda _1 }$. The product of the equation (\ref{eq31}) by its ``conjugated'' equation reads:

\begin{equation}
\label{eq33}
( \vec V \wedge \vec{Y}_{\lambda _1 } ) \cdot ( \vec V \wedge \vec{Y}_{\lambda _2 } ) = 0.
\end{equation}

\smallskip

It has been proved in Ginoux \textit{et al.} \cite{GiRo1} that:

\begin{equation}
\label{eq34}
( \vec V \wedge \vec{Y}_{\lambda _1 } ) \cdot ( \vec V \wedge \vec{Y}_{\lambda _2 } ) = \left\| \vec \gamma \wedge \vec V \right\|^2 = 0,
\end{equation}

\smallskip

where $\vec \gamma = \dot{\vec V} = \ddot{\vec X}$ represents the acceleration vector field, i.e., the time derivative of the velocity vector field (\ref{eq23}). But, in the framework of \textit{Differential Geometry}, \textit{curvature} of trajectory curve integral of \textit{slow-fast dynamical system} (\ref{eq23}) is defined by:

\begin{equation}
\label{eq35}
\dfrac{1}{ \mathcal{R} } = \dfrac{\left\| \vec \gamma \wedge \vec V \right\|}{\left\| \vec V  \right\|^3}
\end{equation}

where $\mathcal{R}$ represents the \textit{radius of curvature}. Thus, it followed from this result that the \textit{slow manifold} of a two-dimensional  \textit{slow-fast dynamical system} (\ref{eq23}) can be directly approximated from the \textit{collinearity condition} between the velocity vector field $\vec V = \dot{\vec X}$ (\ref{eq23}) and its acceleration vector field $\vec \gamma = \dot{\vec V} = \ddot{\vec X}$. More than ten years ago, Ginoux \textit{et al.} \cite{GiRo2} proved that this \textit{collinearity condition} can be written as:

\begin{equation}
\label{eq36}
\det ( \dot{\vec X}, \ddot{\vec X} ) = 0.
\end{equation}

This led Ginoux \textit{et al.} \cite{GiRo2} to the following proposition:

\begin{proposition}
The location of the points where the local curvature of the trajectory curves integral of a two-dimensional dynamical system defined by (\ref{eq23}) vanishes, directly provides the slow manifold analytical equation associated to this system.
\end{proposition}

\smallskip

For a three-dimensional \textit{slow-fast dynamical system} (\ref{eq23}), equation (\ref{eq32}) must be multiplied by two ``conjugated'' equations obtained by
circular shifts of the eigenvalues. The product of equation (\ref{eq32}) by its ``conjugated'' equation reads:

\begin{equation}
\label{eq37}
\left[ \vec V \cdot ( \vec{Y}_{\lambda_1} \wedge \vec{Y}_{\lambda_2} ) \right]
\cdot \left[ \vec V \cdot ( \vec{Y}_{\lambda_2} \wedge \vec{Y}_{\lambda_3} ) \right]
\cdot \left[ \vec V \cdot ( \vec{Y}_{\lambda_1} \wedge \vec{Y}_{\lambda_3} ) \right] = 0.
\end{equation}

\smallskip

It has been proved in Ginoux \textit{et al.} \cite{GiRo1} that:

\begin{equation}
\label{eq38}
\left[ \vec V \cdot ( \vec{Y}_{\lambda_1} \wedge \vec{Y}_{\lambda_2} ) \right]
\cdot \left[ \vec V \cdot ( \vec{Y}_{\lambda_2} \wedge \vec{Y}_{\lambda_3} ) \right]
\cdot \left[ \vec V \cdot ( \vec{Y}_{\lambda_1} \wedge \vec{Y}_{\lambda_3} ) \right] =
\dot{\vec \gamma} \cdot (\vec \gamma \wedge \vec V ) = 0,
\end{equation}

\smallskip

where $\dot{\vec \gamma}$ represents the over-acceleration vector field. But, in the framework of \textit{Differential Geometry}, \textit{torsion} of trajectory curve integral of \textit{slow-fast dynamical system} (\ref{eq23}) is defined by:

\begin{equation}
\label{eq39}
\frac{1}{ \mathcal{T} } = - \frac{\dot {\vec {\gamma }} \cdot \left( {\vec \gamma \wedge \vec V } \right)}{\left\| {\vec \gamma \wedge \vec V }\right\|^2},
\end{equation}

where $\mathcal{T}$ represents the \textit{radius of} \textit{torsion}. Thus, it followed from this result that the \textit{slow manifold} of a three-dimensional  \textit{slow-fast dynamical system} (\ref{eq23}) can be directly approximated from the \textit{coplanarity condition} between the velocity vector field $\vec V = \dot{\vec X}$ (\ref{eq23}), its acceleration vector field $\vec \gamma$ and its over-acceleration $\dot{\vec \gamma}$. More than ten years ago, Ginoux \textit{et al.} \cite{GiRo2} proved that this \textit{coplanarity condition} can be written as:

\begin{equation}
\label{eq40}
\det ( \dot{\vec X}, \ddot{\vec X}, \dddot{\vec X} ) = 0.
\end{equation}

\smallskip

This led Ginoux \textit{et al.} \cite{GiRo2} to the following proposition:

\begin{proposition}
The location of the points where the local torsion of the trajectory curves integral of a three-dimensional dynamical system defined by (\ref{eq23}) vanishes, directly provides the slow manifold analytical equation associated to this system.
\end{proposition}

\smallskip

Then, Ginoux \textit{et al.} \cite{GiRo2} and Ginoux \cite{Gin} generalized these results to any $n$-dimensional \textit{slow-fast dynamical system} (\ref{eq23}) as well as any $n$-dimensional \textit{singularly perturbed dynamical system} (\ref{eq1}). This led to the following general proposition which encompasses proposition 2 \& 3:

\begin{proposition}
The location of the points where the curvature of the flow, i.e., the curvature of the trajectory curves of any $n$-dimensional dynamical system (\ref{eq1}) or (\ref{eq23}) vanishes, directly provides its $(n - 1)$-dimensional slow invariant manifold analytical equation which reads:

\begin{equation}
\label{eq41} \phi ( \vec {X} ) = \dot {\vec {X}} \cdot ( \ddot {\vec {X}} \wedge \dddot {\vec {X}} \wedge \ldots \wedge \mathop {\vec {X}}\limits^{\left( n \right)} ) = det ( \dot {\vec {X}},\ddot {\vec {X}},\dddot {\vec {X}},\ldots,\mathop {\vec {X}}\limits^{\left( n \right)}  ) = 0.
\end{equation}

\end{proposition}

\begin{proof}
For proof of this proposition 4, see Ginoux \textit{et al.} \cite{GiRo2} and Ginoux \cite{Gin}.
\end{proof}

Let's notice that within the framework of \textit{Differential Geometry}, $n$-dimensional \textit{smooth curves}, i.e., \textit{smooth curves} in Euclidean $n-$space are defined by a \textit{regular parametric representation} in \textit{terms of arc length} also called \textit{natural representation} or \textit{unit speed parametrization}. According to Herman Gluck \cite{Gluck} local metrics properties of \textit{curvatures} may be directly deduced from \textit{curves parametrized in terms of time} and so \textit{natural representation} is not necessary. This fundamental result allows applying the \textit{Grahm-Schmidt orthogonalization process}, with a collection of $i = 1, \ldots, n$ vectors $\vec{u}_i(t)$ depending explicitly on \textit{time} (not on \textit{arc length}) and forming an orthogonal basis, to deduce the expression of \textit{curvatures} of \textit{trajectory curves} integral of \textit{dynamical systems} (\ref{eq1}) or (\ref{eq23}) in Euclidean $n-$space.\\

Then, to prove that the \textit{slow manifold} (\ref{eq41}) is locally invariant by the flow of (\ref{eq23}), we use \textit{Darboux invariance theorem}. According to Schlomiuk \cite{Schlomiuk} and Llibre \textit{et al.} \cite{LLibreMedrado} it seems that Gaston Darboux \cite{Darboux} has been the first to define the concept of \textit{invariant manifold} which can be presented as follows:

\begin{proposition}

The \textit{manifold} defined by $\phi ( \vec {X} ) = 0$ where $\phi $ is a $C^1$ in an open set U is \textit{invariant} with respect to the flow of (\ref{eq1}) or (\ref{eq23}) if there exists a $C^1$ function denoted $k( \vec {X} )$ and called cofactor which satisfies:

\begin{equation}
\label{eq42}
L_{\vec V } \phi ( \vec{X} ) = k ( \vec {X} ) \phi ( \vec {X} ),
\end{equation}

for all $\vec {X} \in U$ and with the Lie derivative operator defined as:

\[
L_{\vec V } \phi = \vec V \cdot \vec \nabla \phi = \sum\limits_{i = 1}^n {\frac{\partial \phi }{\partial x_i }\dot {x}_i } = \frac{d\phi }{dt}.
\]

\end{proposition}

\begin{proof} To state the invariance of \textit{slow manifold} (\ref{eq42}), it is first necessary to recall the following identities (\ref{eq43})-(\ref{eq44}) the proof of which is obvious:

\begin{equation}
\label{eq43}
\ddot {\vec {X}} = J \dot {\vec {X}},
\end{equation}

where $J$ is the Jacobian matrix associated with any $n$-dimensional \textit{slow-fast dynamical system} (\ref{eq23}).

\begin{equation}
\label{eq44}
J\vec {a}_1 .\left( {\vec {a}_2 \wedge \ldots \wedge \vec {a}_n }
\right)+\vec {a}_1 .\left( {J\vec {a}_2 \wedge \ldots \wedge \vec {a}_n }
\right) +\ldots  +\vec {a}_1 .\left( {\vec {a}_2 \wedge \ldots \wedge J\vec
{a}_n } \right) = Tr\left( J \right)\vec {a}_1 .\left( {\vec {a}_2 \wedge
\ldots \wedge \vec {a}_n } \right).
\end{equation}

Thus, the Lie derivative of the \textit{slow manifold} (\ref{eq42}) reads:

\begin{equation}
\label{eq45}
L_{\vec V } \phi ( {\vec {X}} )=\dot {\vec {X}}\cdot ( {\ddot {\vec {X}} \wedge \dddot{\vec {X}}\wedge \ldots \wedge \mathop {\vec {X}}\limits^{\left( {n+1} \right)} } ),
\end{equation}

From the identity (\ref{eq43}) we find that:

\begin{equation}
\label{eq46}
\mathop {\vec {X}}\limits^{\left( {n+1} \right)} = J^n\dot {\vec {X}} \mbox{\quad if \quad} \frac{dJ}{dt}=0,
\end{equation}

where $J^n$ represents the $n^{th}$ power of $J$. Then, it follows that:

\begin{equation}
\label{eq47}
\mathop {\vec {X}}\limits^{\left( {n+1} \right)} = J\mathop {\vec {X}}\limits^{\left( n \right)}.
\end{equation}

Replacing $\mathop {\vec {X}}\limits^{\left( {n+1} \right)}$ in Eq. (\ref{eq45}) with Eq. (\ref{eq47}) we have:

\begin{equation}
\label{eq48}
L_{\vec V } \phi ( {\vec {X}} ) = \dot {\vec {X}}\cdot ( {\ddot {\vec {X}}\wedge \dddot{\vec {X}}\wedge \ldots \wedge J\mathop {\vec{X}}\limits^{\left( n \right)} } ).
\end{equation}

The right hand side of this Eq. (\ref{eq48}) can be written:

\[
J \dot {\vec {X}}\cdot ( {\ddot {\vec {X}}\wedge \dddot{\vec {X}}\wedge \ldots \wedge \mathop {\vec{X}}\limits^{\left( n \right)} } ) + \dot {\vec {X}}\cdot ( J {\ddot {\vec {X}}\wedge \dddot{\vec {X}}\wedge \ldots \wedge \mathop {\vec{X}}\limits^{\left( n \right)} } ) + \ldots + \dot {\vec {X}}\cdot ( {\ddot {\vec {X}}\wedge \dddot{\vec {X}}\wedge \ldots \wedge J\mathop {\vec{X}}\limits^{\left( n \right)} } ).
\]

According to Eq. (\ref{eq47}) all terms are null except the last one. So, by taking into account identity (\ref{eq44}), we find:

\[
L_{\vec V } \phi ( {\vec {X}} ) = Tr\left( J \right) \dot {\vec {X}}\cdot ( {\ddot {\vec {X}}\wedge \dddot{\vec {X}}\wedge \ldots
\wedge \mathop {\vec {X}}\limits^{\left( n \right)} } ) = Tr\left( J \right)\phi ( {\vec {X}} ) = k( {\vec {X}} ) \phi ( {\vec {X}} ),
\]

where $k( {\vec {X}} ) = Tr\left( J \right)$ represents the trace of the Jacobian matrix.\\

So, the invariance of the \textit{slow manifold} analytical equation of any $n$-dimensional \textit{slow-fast dynamical system} (\ref{eq23}) is established provided that the Jacobian matrix is locally stationary. \end{proof}

More than ten years ago, Ginoux \cite{Gin} established for any $n$-dimensional \textit{singularly perturbed dynamical systems} (\ref{eq1}) the identity between Fenichel's invariance (\ref{eq5}) and Darboux invariance (\ref{eq42}) (see also Ginoux \cite{GinLi3}).\\

Thus, it follows from what precedes that both \textit{Tangent Linear System Approximation Method} and \textit{Flow Curvature Method} are absolutely identical.

\subsection{Inflection Line Method}

In 1994, Br{\o}ns and Bar-Eli \cite{Brons1994} presented their \textit{Inflection Line Method} for planar or two-dimensional \textit{slow-fast dynamical system} (\ref{eq23}) which can be written as follows:

\begin{equation}
\label{eq49}
\left\{ \begin{aligned}
\varepsilon \dot {x} & = f\left( x, y, \varepsilon \right), \hfill \\
 \dot {y} & = g\left( x,y, \varepsilon  \right). \hfill \\
\end{aligned}  \right.
\end{equation}

According to Br{\o}ns and Bar-Eli \cite{Brons1994}:

\begin{quote}
``The curvature $\kappa$ of a smooth curve measures the rate of turn of the tangent vector with respect to arc length $s$. The \textit{inflection line} $l_i$ associated with (\ref{eq49}) is the locus of points where $\kappa_t = 0$. The inflection line typically consists of one or more smooth curves. For our purposes, it suffices to consider systems (\ref{eq49}) in a region where $f(x, y, \varepsilon) \neq 0$. Then time can be eliminated to give the equivalent equation for the trajectories,

\begin{equation}
\label{eq50}
\dfrac{dy}{dx} = \dfrac{g\left( x, y, \varepsilon \right)}{f\left( x, y, \varepsilon \right)}.
\end{equation}

For curves in the ($x,y$)-plane of the form $y = y(x)$, the curvature can be computed as

\begin{equation}
\label{eq51}
\kappa = \dfrac{y''(x)}{\left(1 + y'(x) \right)^{3/2}},
\end{equation}

where the positive root in the denominator is used. Thus, the inflection line for (\ref{eq50}) is determined by $y''(x) = 0$, and an equation for it can be found by differentiating (\ref{eq50}).''

\end{quote}

\smallskip

For such two-dimensional dynamical systems (\ref{eq49}), the issue of the parametrization highlighted by Gluck \cite{Gluck} just above can be easily solved by dividing the two components of the velocity vector field (\ref{eq49}) as done by Br{\o}ns and Bar-Eli \cite{Brons1994}. Nevertheless, such simplification is no more possible in dimension higher than two. Let's notice that for three-dimensional \textit{dynamical systems} the use of \textit{curvature} instead of \textit{torsion} has already been carried out by Ginoux \textit{et al.} \cite{GiRo2a} and by Gilmore \textit{et al.} \cite{Gilmore}. However, the use of \textit{curvature} or more precisely, the \textit{first curvature} require to solve a system of two nonlinear equations since in dimension three a \textit{curve} is defined as the intersection of two surfaces. Nowadays, this kind of problem cannot be solved analytically in the general case. Thus, such considerations precludes any generalization of the \textit{Inflection Line Method} proposed by Br{\o}ns and Bar-Eli \cite{Brons1994} at least from an analytical point of view.

Now let's prove that the \textit{Inflection Line Method} is just a particular case of the more general \textit{Flow Curvature Method}. According to Br{\o}ns and Bar-Eli \cite{Brons1994} ``the inflection line for (\ref{eq50}) is determined by $y''(x) = 0$, and an equation for it can be found by differentiating (\ref{eq50}).'' So, let's differentiate equation (\ref{eq50}):

\begin{equation}
\label{eq52}
\dfrac{d}{dx}\left( \dfrac{dy}{dx} \right) = \dfrac{d}{dx}\left( \dfrac{dy}{dt} \dfrac{dt}{dx} \right) = \dfrac{d}{dt}\left( \dfrac{\dot{y}}{\dot{x}} \right) \dfrac{dt}{dx} = \dfrac{\ddot{y}\dot{x} - \ddot{x}\dot{y}}{\dot{x}^3} = \dfrac{1}{\dot{x}^3} \begin{vmatrix} \dot{x} & \ddot{x} \\ \dot{y} & \ddot{y} \end{vmatrix} = 0.
\end{equation}

So, the condition $y''(x) = 0$ given by Eq. (\ref{eq52}) is absolutely identical to the condition $\det ( \dot{\vec X}, \ddot{\vec X} ) = 0$ where $\vec X = (x, y)^t$ provided by (\ref{eq36}). As a consequence, it is thus proved that the \textit{Inflection Line Method} is just a particular case of the more general \textit{Flow Curvature Method}.

To resume, we have proved in this section that the \textit{Flow Curvature Method} encompasses the \textit{Intrinsic Low-Dimensional Manifold Method}, the \textit{Tangent Linear System Approximation Method} and the \textit{Inflection Line Method} which all belong to the second category called \textit{Curvature-Based Methods}.

Now, a last issue arises as follows: for $n$-dimensional \textit{singularly perturbed dynamical systems} (\ref{eq1}), does the \textit{Flow Curvature Method} provide the same approximation of their \textit{slow invariant manifold} as the \textit{Geometric Singular Perturbation Method}?

\section{Singular Perturbation and Curvature Based Methods}

While the \textit{slow invariant manifold} analytical equation (\ref{eq6}) given by the \textit{Geometric Singular Perturbation Theory} is an \textit{explicit} equation, the \textit{slow invariant manifold} analytical equation (\ref{eq41}) obtained according to the \textit{Flow Curvature Method} is an \textit{implicit} equation. So, in order to compare the latter with the former it is necessary to plug the following perturbation expansion: $\vec {Y}\left( {\vec {x},\varepsilon } \right) = \vec {Y}_0 \left( \vec {x} \right) + \varepsilon \vec {Y}_1 \left( \vec {x} \right) + O\left( {\varepsilon ^2} \right)$ into (\ref{eq41}). Thus, solving order by order for $\vec {Y}\left( {\vec {x},\varepsilon } \right)$ will transform (\ref{eq41}) into an \textit{explicit analytical equation} enabling the comparison with (\ref{eq6}). The Taylor series expansion for $\phi (\vec{X}, \varepsilon ) = \phi (\vec{x}, \vec {Y}\left( {\vec {x},\varepsilon } \right), \varepsilon )$ up to terms of order one in $\varepsilon $ reads:

\begin{equation}
\label{eq53}
\phi (\vec{X}, \varepsilon ) = \phi (\vec{x}, \vec {Y}_0
\left( \vec {x} \right),0 ) + \varepsilon D_{\vec{y}} \phi (\vec{x}, \vec {Y}_0 \left( \vec {x} \right),0 )\vec {Y}_1 \left( \vec {x} \right) + \varepsilon \frac{\partial \phi }{\partial \varepsilon } (\vec{x}, \vec {Y}_0 \left( \vec {x} \right),0 ).
\end{equation}

\textbullet \hspace{0.1in} At order $\varepsilon ^0$, Eq. (\ref{eq53}) gives:

\begin{equation}
\label{eq54}
\phi \left(\vec {x}, \vec {Y}_0 \left( \vec {x} \right),0 \right) = 0,
\end{equation}

which defines $\vec {Y}_0 \left( \vec {x} \right)$ due to the invertibility of $D_{\vec {y}} \phi $ and application of the \textit{Implicit Function Theorem}.\\

\textbullet \hspace{0.1in} The next order $\varepsilon ^1$, provides:

\begin{equation}
\label{eq55}
D_{\vec{y}} \phi (\vec{x}, \vec {Y}_0 \left( \vec {x} \right),0 )\vec {Y}_1 \left( \vec {x} \right) + \frac{\partial \phi }{\partial \varepsilon } (\vec{x}, \vec {Y}_0 \left( \vec {x} \right),0 ) = \vec {0},
\end{equation}

which yields $\vec {Y}_1 \left( \vec {x} \right)$ and so forth.\\

In order to prove that this equation is completely identical to Eq. (\ref{eq8}), let's rewrite it as follows:

\[
\vec {Y}_1 \left( \vec {x} \right) = - \left[ D_{\vec{y}} \phi (\vec{x}, \vec {Y}_0 \left( \vec {x} \right),0 ) \right]^{-1} \frac{\partial \phi }{\partial \varepsilon } (\vec{x}, \vec {Y}_0 \left( \vec {x} \right),0 ).
\]

By application of the \textit{chain rule}, i.e., the derivative of $\phi (\vec{x}, \vec {Y}_0 \left( \vec {x} \right),0)$ with respect to the variable $\vec {y}$ and then with respect to $\varepsilon$, it can be stated that:

\[
\vec {Y}_1 \left( \vec {x} \right)  = - \left[ (D_{\vec {x}} \vec {f}) (D_{\vec {y}}\vec {f}) \right]^{-1} (D_{\vec {y}}\vec {f}) \vec {g}(\vec{x}, \vec {Y}_0 ( \vec {x} ),0 ) - \left[ D_{\vec {y}}\vec {f} \right]^{-1} D_{\varepsilon}\vec {f}(\vec{x}, \vec {Y}_0 \left( \vec {x} \right),0 ).
\]

But, according to the \textit{Implicit Function Theorem} we have:

\[
(D_{\vec {x}} \vec {f}) = - (D_{\vec {y}} \vec {f}) (D_{\vec {x}} \vec {y}) = - (D_{\vec {y}} \vec {f}) (D_{\vec {x}} \vec {Y}_0 ( \vec {x} )).
\]

Then, by replacing into the previous equation we find:

\[
\vec {Y}_1 \left( \vec {x} \right)  = \left[ (D_{\vec {y}} \vec {f}) (D_{\vec {x}} \vec {Y}_0 ( \vec {x} )) (D_{\vec {y}}\vec {f}) \right]^{-1} (D_{\vec {y}}\vec {f}) \vec {g}(\vec{x}, \vec {Y}_0 ( \vec {x} ),0 ) - \left[ D_{\vec {y}}\vec {f} \right]^{-1} D_{\varepsilon}\vec {f}(\vec{x}, \vec {Y}_0 \left( \vec {x} \right),0 ).
\]

After simplifications, we have:

\[
\vec{Y_1}\left( \vec{x} \right) = \left[ D_{\vec {x}} \vec {Y}_0 \left( \vec {x} \right) D_{\vec {y}} \vec {f} (\vec{x}, {\vec {Y}_0 \left( \vec {x} \right),0} ) \right]^{-1} \vec {g}(\vec{x}, {\vec {Y}_0 \left( \vec {x} \right),0} ) - \left[ D_{\vec {y}} \vec {f} (\vec{x}, {\vec {Y}_0 \left( \vec {x} \right),0} ) \right]^{-1} \frac{\partial \vec {f}}{\partial \varepsilon }.
\]

Finally, Eq. (\ref{eq55}) may be written as:

\[
D_{\vec {x}} \vec {Y}_0 \left( \vec {x} \right) \left[ D_{\vec {y}} \vec {f} (\vec{x}, {\vec {Y}_0 \left( \vec {x} \right),0} ) \vec{Y_1}\left( \vec{x} \right) + \frac{\partial \vec {f}}{\partial \varepsilon } \right] = \vec {g}(\vec{x}, {\vec {Y}_0 \left( \vec {x} \right),0} ).
\]

Thus, identity between the \textit{slow invariant manifold} equation given by the \textit{Geometric Singular Perturbation Theory} and by the \textit{Flow Curvature Method} is proved up to first order term in $\varepsilon$.\\

Let's notice that the \textit{slow invariant manifold} equation (\ref{eq41}) associated with $n$--dimensional \textit{singularly perturbed systems} defined by the \textit{Flow Curvature Method} is a tensor of order $n$. As a consequence, it can only provide an approximation of $n$-order in $\varepsilon$ of the \textit{slow invariant manifold} equation (\ref{eq6}). In order to emphasize this result let's come back to the example of the Van der Pol system \cite{Vdp1926}.

\subsection{Van der Pol system}

Let's consider the Van der Pol system (\ref{eq9}) which is a well-known two-dimensional \textit{singularly perturbed dynamical system}. By application of the \textit{Geometric Singular Perturbation Method} an approximation up to order $\e^3$ of the \textit{slow invariant manifold} of the Van der Pol \textit{singularly perturbed dynamical system} (\ref{eq9}) is given by Eq. (\ref{eq13}):

\begin{equation}
\label{eq56}
y = Y(x, \e) = \dfrac{x^3}{3} - x + \e \dfrac{x}{1-x^2}  + \e^2 \dfrac{x (1 + x^2)}{(1 - x^2)^4} + O\left( {\varepsilon^3} \right).
\end{equation}

By using the \textit{Flow Curvature Method} the \textit{slow manifold} equation associated with this system reads:

\begin{equation}
\label{eq57}
\phi ( {\vec {X}} )=\phi \left( {x,y,\varepsilon } \right)=9y^2+\left( {9x+3x^3} \right)y+6x^4-2x^6+9x^2\varepsilon = 0.
\end{equation}

So, by plugging the perturbation expansion (\ref{eq56}) into (\ref{eq57}) and then, solving order by order leads to:

\begin{equation}
\label{eq58}
y=\frac{x^3}{3}-x+\varepsilon \frac{x}{1-x^2} + \varepsilon ^2\frac{x}{\left({1 - x^2} \right)^3}+O\left( {\varepsilon ^3} \right).
\end{equation}

\begin{figure}[htbp]
  \begin{center}
    \begin{tabular}{ccc}
      \includegraphics[width=8cm,height=8cm]{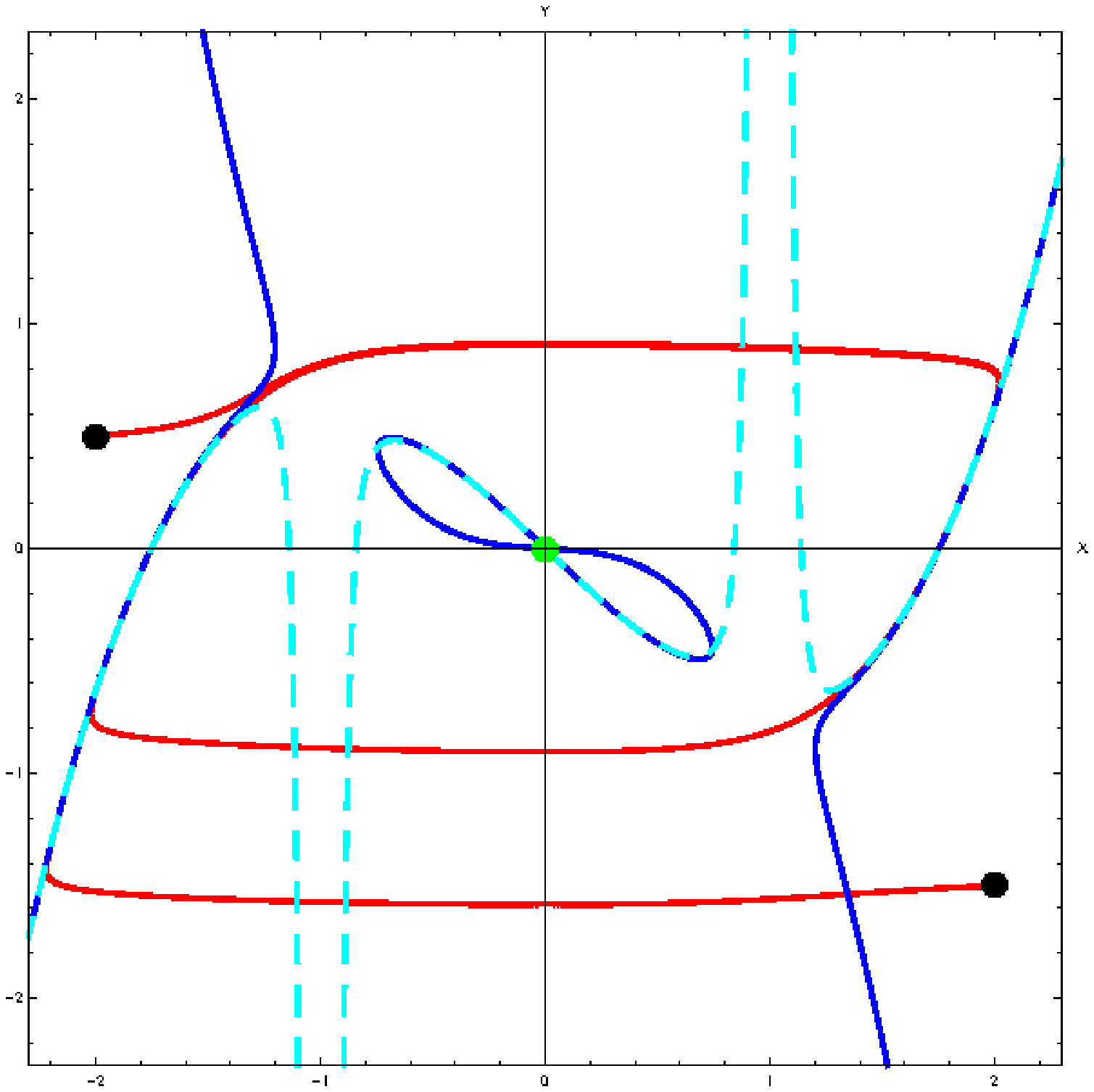} & ~~~~~~ &
      \includegraphics[width=8cm,height=8cm]{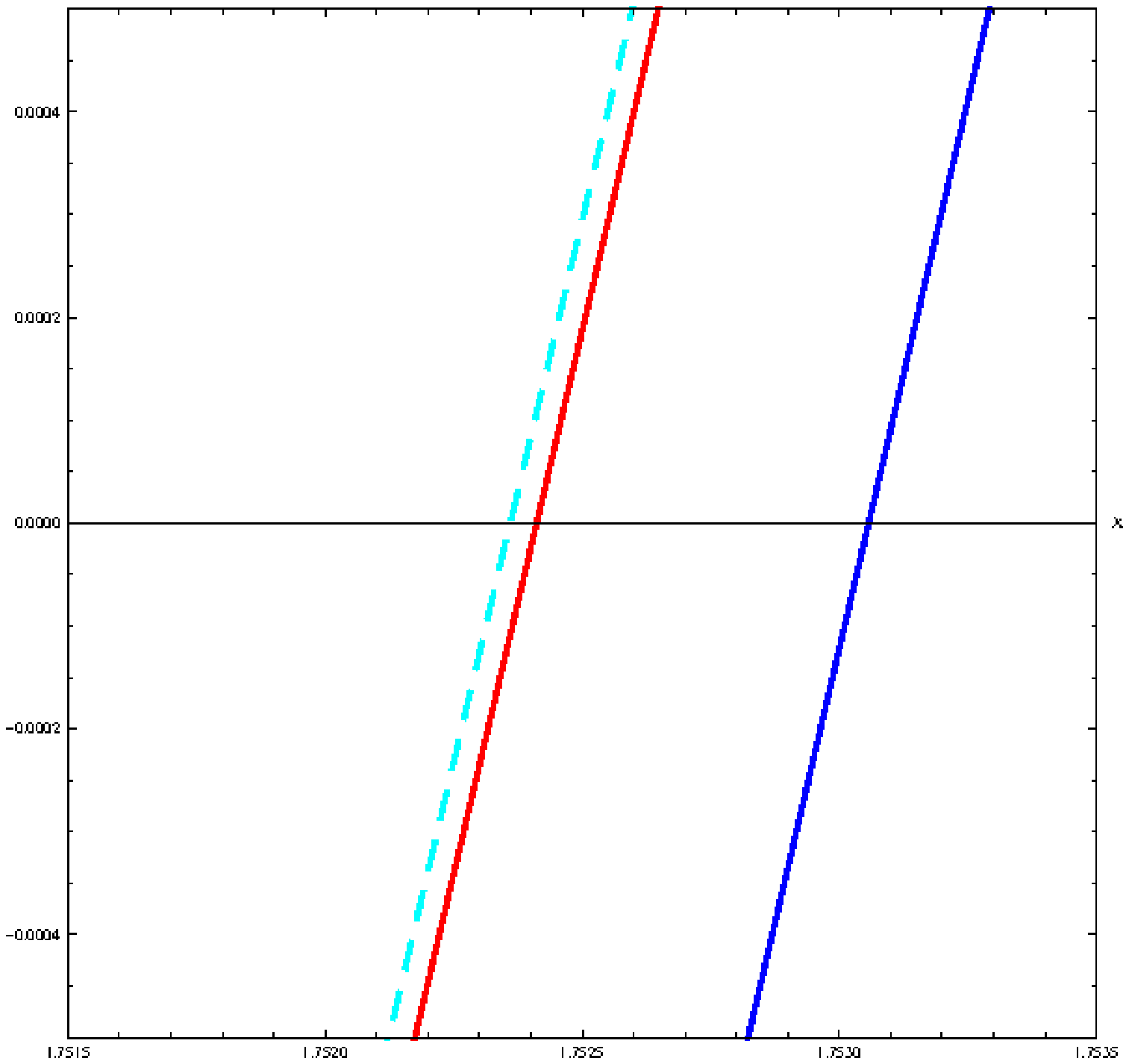} \\
      (a) & & (b) \\[0.2cm]
      \includegraphics[width=8cm,height=8cm]{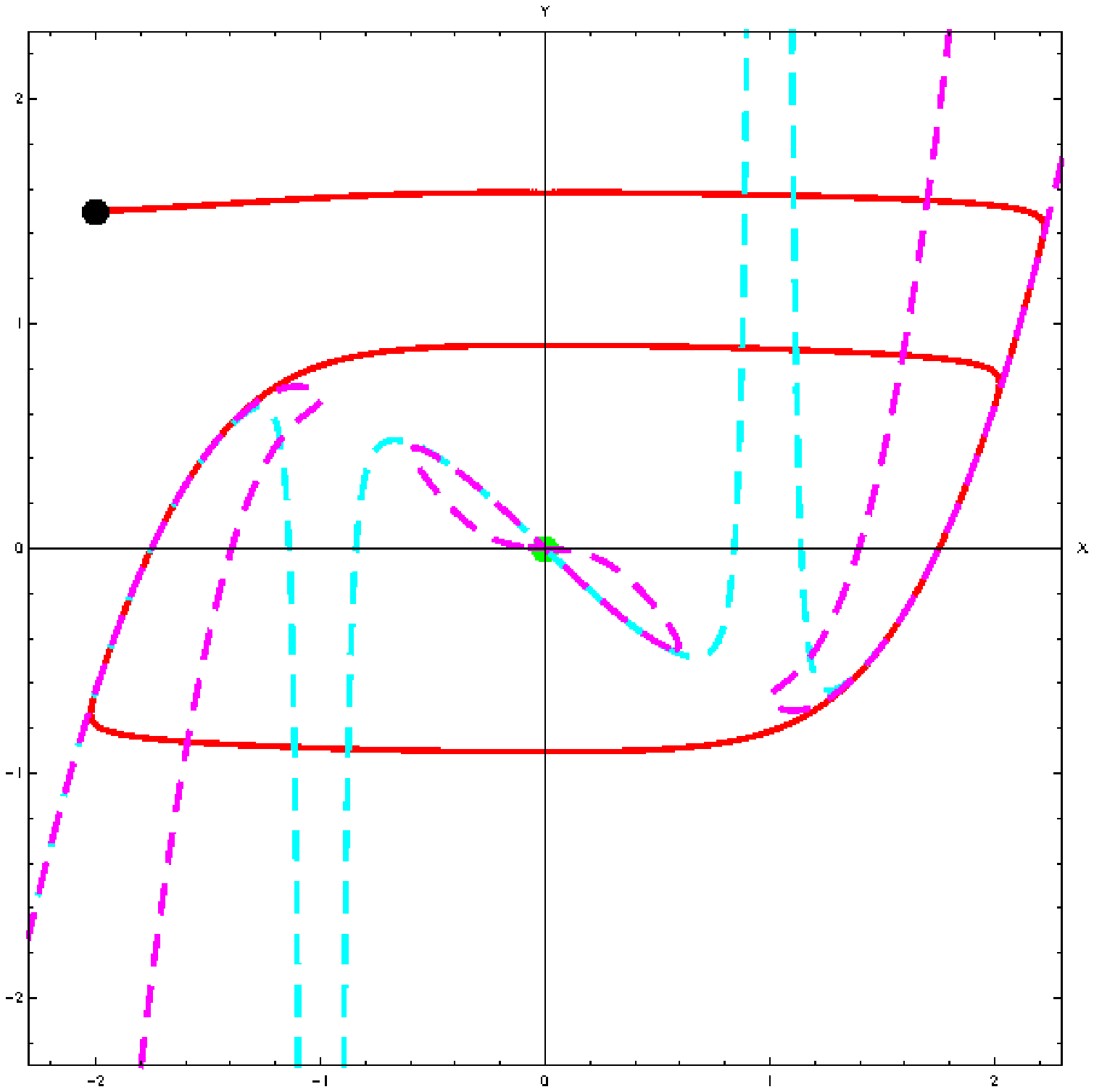} & ~~~ &
      \includegraphics[width=8cm,height=8cm]{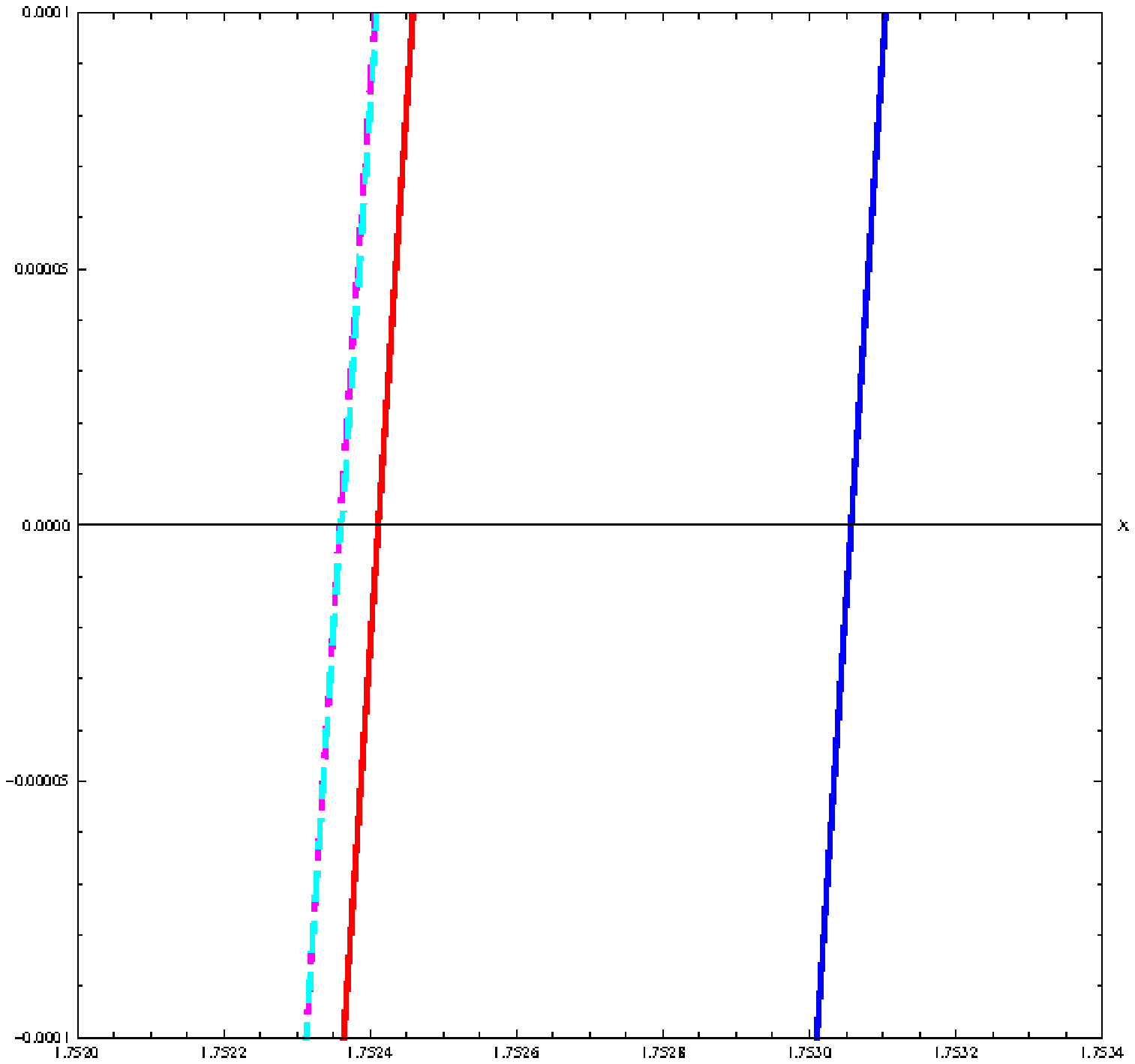} \\
      (c) & & (d) \\[-0.2cm]
    \end{tabular}
    \caption{Slow manifold of Van der Pol system (\ref{eq9}).}
    \label{fig1}
  \end{center}
  \vspace{-0.5cm}
\end{figure}

\newpage

So, both \textit{slow invariant manifolds} approximated by Eqs. (\ref{eq56}) \& (\ref{eq58}) are completely identical up to order one in $\varepsilon $. At order $\varepsilon^2$ a difference appears which is due to the fact that the \textit{slow invariant manifold} obtained with the \textit{Flow Curvature Method} is defined, for a two-dimensional dynamical system, by the \textit{second order tensor} of \textit{curvature}, i.e. by a determinant (\ref{eq36}) involving the first and second time derivatives of $\vec X$. If one makes the same computation as previously but with the Lie derivative of this determinant we obtain a determinant containing the first and third time derivatives of $\vec X$, i.e. the \textit{third order tensor} $\det ( \dot{\vec X}, \dddot{\vec X} ) = 0$. Then there is no more difference between order two in $\varepsilon$ and \textit{slow invariant manifold} given by both methods are exactly the same. These results, highlighted in figures 1 and more particularly in Fig. 1d, had been already found by Bruno Rossetto \cite{Rossetto1986, Rossetto1987} by using his \textit{Successive Approximations Method} and then, by Gear \textit{et al.} \cite{Gear} and Zagaris \textit{et al.} \cite{Zagaris} with their \textit{Zero-Derivative Principle}. In figures 1, we have plotted in blue the \textit{slow invariant manifold} approximated with the \textit{Flow Curvature Method} (see Fig. 1a) and in cyan dotted line the \textit{slow invariant manifold} approximated with the \textit{Geometric Singular Perturbation Method} (see Fig. 1a \& 1c). Both \textit{slow invariant manifolds} are compared in figure 1b in which we emphasize that \textit{Flow Curvature Method} only provides an approximation up to order one in $\varepsilon $. In figures 1c \& 1d, we have plotted the time derivative of the \textit{slow invariant manifold} obtained with the \textit{Flow Curvature Method}, i.e., given by the \textit{third order tensor} $\det ( \dot{\vec X}, \dddot{\vec X} ) = 0$ in magenta dotted  line. Then, as observed on Fig. 1d, there is no difference between both \textit{slow invariant manifolds}. According to Beno\^{i}t \textit{et al.} \cite{Benoit2015}:

\begin{quote}
``However, the global solution sets are very complicated and have branches which do not approximate any slow manifold. These branches are the ghosts. Close to
the fold points of the critical manifolds the branches which approximate a Fenichel manifold merge with a ghost in a fold. Furthermore, we see that the solution to $G = 0$ [$\phi ( {\vec {X}} )= 0$, according to the \textit{Flow Curvature Method}] has a double point at the equilibrium where a ghost intersects the approximation to the Fenichel manifold transversally.''
\end{quote}

As highlighted in Fig. 1a \& 1c, the \textit{slow invariant manifold} approximated with the \textit{Flow Curvature Method} obviously contains some ``ghost parts'' near the \textit{fold}. Nevertheless, let's notice that such ``ghost parts'' do also exist for the \textit{Geometric Singular Perturbation Theory} as emphasized in Fig. 1a \& 1c by the dark points representing various initial conditions. Moreover, the ``double point at the equilibrium where a ghost intersects the approximation to the Fenichel manifold transversally'', i.e. the curve with the shape of a lemniscate of Bernoulli, provides in fact the eigendirections of both \textit{slow} and \textit{fast} eigenvectors of Van der Pol system at the equilibrium point, i.e. the origin.

\subsection{Lorenz system}

In the beginning of the sixties a young meteorologist of the M.I.T., Edward N. Lorenz, was working on weather prediction. He elaborated a
model derived from the Navier-Stokes equations with the Boussinesq approximation, which described the atmospheric convection. Although his
model lost the correspondence to the actual atmosphere in the process of approximation, chaos appeared from the equation describing the dynamics
of the nature. Let's consider the Lorenz model \cite{Lorenz}:

\begin{equation}
\label{eq59}
\left\{
\begin{aligned}
 \dot {x} & = \sigma \left( y - x \right), \hfill \\
 \dot {y} & = -x z + r x -  y, \hfill \\
 \dot {z} & = x y - b z, \hfill \\
\end{aligned}
\right.
\end{equation}

with the following parameters set $\sigma = 10$, $b = 8/3$ and by setting $r = 28$ this \textit{dynamical system} exhibits the famous Lorenz butterfly. Although the Lorenz's system (\ref{eq59}) has no \textit{singular approximation} or no \textit{critical manifold}, it has been numerically stated by Rossetto \textit{et al.} \cite{Rossetto1998} that its Jacobian matrix possesses at least a large and negative real eigenvalue in a large domain of the phase space. So, it is considered as a \textit{slow fast dynamical system} but not as a \textit{singularly perturbed dynamical system}. Thus, none of the \textit{singular perturbation-based methods} can provide the \textit{slow invariant manifold} associated with Lorenz system. However, it can be obtained with any \textit{curvature-based methods} (see Rossetto \textit{et al.} \cite{Rossetto1998}) and more particularly with the \textit{Flow Curvature Method} (see Ginoux \textit{et al.} \cite{GiRo1} and Ginoux \cite{Gin}) which gives the following equation:

\begin{align}
\label{eq60}
&\phi(x,y,z) = -12954816 x^4 + 21168 x^6 + 13331304 x^3 y - 2772 x^5 y \hfill \notag \\
&- 15012 x^4 y^2 + 27 x^6 y^2 - 6554520 x y^3 + 13410 x^3 y^3 + 340200 y^4 \hfill \notag \\
& + 15120 x^2 y^4 - 8100 x y^5 +15906240 x^2 z + 1311744 x^4 z - 1512 x^6 z \hfill \notag \\
& + 5112720 x y z -1454430 x^3 y z- 45 x^5 y z - 5680800 y^2 z - 720 x^2 z^4 \hfill \notag \\
& - 790440 x^2 y^2 z + 540 x^4 y^2 z + 456180 x y^3 z - 810 x^3 y^3 z + 1800 y^4 z\hfill \notag \\
& -1686544 x^2 z^2 - 45252 x^4 z^2 + 27 x^6 z^2 - 317920 x y z^2 +58320 x^3 y z^2\hfill \notag \\
& + 372800 y^2 z^2 + 15750 x^2 y^2 z^2 -8100 x y^3 z^2 + 59040 x^2 z^3 + 540 x^4 z^3\hfill \notag \\
&+ 9297666 x^2 y^2 + 6480 x y z^3 -810 x^3 y z^3 - 7200 y^2 z^3 = 0. \hfill
\end{align}

In fact, it can been numerically stated that in the vicinity of this \textit{slow manifold}, the Jacobian matrix is quasi-stationary and so $dJ/dt = 0$. So, the invariance of this \textit{slow manifold} (\ref{eq60}) is stated according to Proposition 5. Let's notice that such \textit{slow invariant manifold} exhibits the symmetry of Lorenz's system (\ref{eq59}), i.e.  $\phi(-x,-y, z) = \phi(x,y,z)$.

\begin{figure}[htbp]
\centerline{\includegraphics[width=16cm,height=16cm]{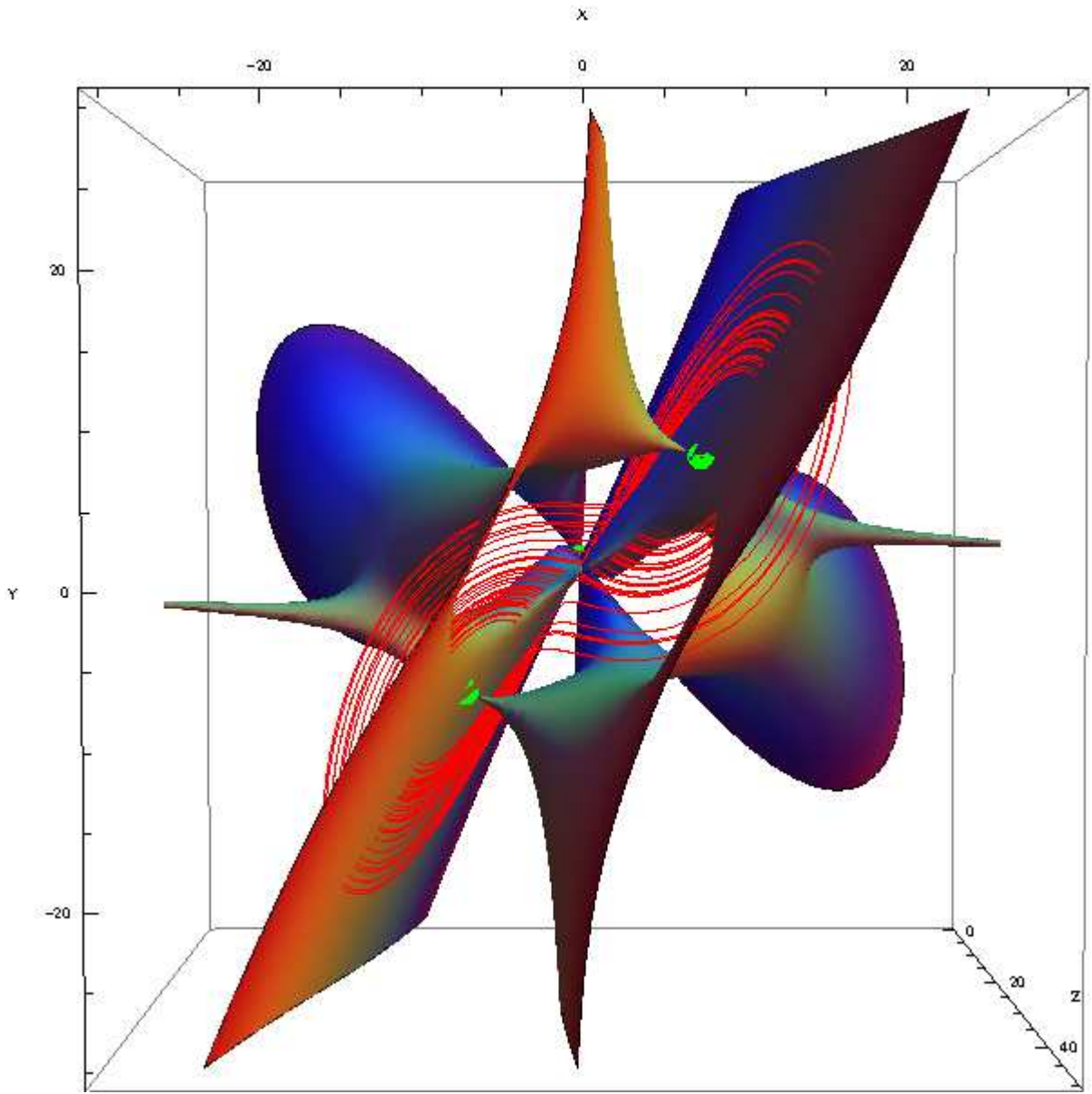}}
\caption{Slow invariant manifold of the Lorenz system \eqref{eq59}.}
\label{fig2}
\end{figure}

\newpage

On figure 2, we observe some kind of ``horns'' near the two fixed points. Such ``horns'' are the three-dimensional analog to the ``ghost parts'' in dimension two. As previously, they provide the eigendirection of both \textit{slow} and \textit{fast} eigenvectors of Lorenz system at the fixed points. For the origin, their shapes correspond to the nature of this equilibrium point, i.e., a \textit{saddle-node}. For the two others, these ``horns'' represent the \textit{cone of attraction} of the eigendirection of the negative real eigenvalue.

Let's notice that Lorenz system is paradigmatic since it also implies chaos in lasers due to the fact that a class C laser is equivalent to the Lorenz model. So, the determination of the \textit{slow invariant manifold} equation of such system is of great importance. Beyond this, class B lasers, ruled only by two equations, are \textit{slow-fast dynamical systems} and the addition of a third equation or a time modulation on a system parameter leads to chaos \cite{Meucci2021}.

\newpage

\section{Discussion}

During the last century, the analysis of a great number of phenomena modeled with dynamical systems, i.e., with sets of nonlinear ordinary differential equations, highlighted the existence of at least two time scales for their evolution: a \textit{slow} time and a \textit{fast} time. This was transcribed by the presence of at least a small multiplicative parameter $\varepsilon$ in the velocity vector field of these dynamical systems. They were thus called \textit{singularly perturbed dynamical systems} and it was proved that they possess \textit{slow invariant manifolds}. Then, various methods were developed to compute such \textit{slow invariant manifolds} or, at least an asymptotic expansion in power of $\e$. Thus, determination of the \textit{slow invariant manifold} equation turned into a regular perturbation problem and the seminal works of Wasow \cite{Wasow}, Cole \cite{Cole}, O'Malley \cite{Malley1, Malley2} and Fenichel \cite{Fen5, Fen6, Fen7, Fen8} gave rise in the 1960s-1970s to the so-called \textit{Geometric Singular Perturbation Theory}. At the end of the 1980s, Rossetto \cite{Rossetto1986, Rossetto1987} developed the \textit{Successive Approximations Method} to approximate the \textit{slow invariant manifold} of \textit{singularly perturbed dynamical systems}. In 2005, Gear \textit{et al.} \cite{Gear} and then, Zagaris \textit{et al.} \cite{Zagaris} used the \textit{Zero-Derivative Principle} for the same purpose. We have established in this work that \textit{Geometric Singular Perturbation Theory}, \textit{Successive Approximations Method} and \textit{Zero-Derivative Principle} are absolutely identical, i.e., provide exactly the same approximation of \textit{slow invariant manifold} equation and so belong to the first category we have called: \textit{Singular Perturbation-Based Methods}. However, according to O'Malley \cite{Malley1}, Rossetto \cite{Rossetto1987} and Beno\^{i}t \textit{et al.} \cite{Benoit2015}, the main drawback of these methods is that the validity of the asymptotic expansion in power of $\e$, approximating the \textit{slow invariant manifold} equation, is expected to breakdown near the \textit{fold} or, near non-hyperbolic regions.\\

Beside, \textit{singularly perturbed dynamical systems} there are \textit{dynamical systems} without an explicit timescale splitting. Some of these systems, which have been called \textit{slow-fast dynamical systems}, have the following property: their Jacobian matrix has at least a \textit{fast eigenvalue}, i.e. with the largest absolute value of the real part. For such systems, \textit{Singular Perturbation-Based Methods} can not be applied anymore. Thus, in the beginning of the 1990s various approaches have been proposed in order to approximate slow manifolds of such \textit{slow-fast dynamical systems}. In 1992, Maas and Pope introduced the \textit{Intrinsic Low-Dimensional Manifold Method} \cite{Maas} and two year later, Br{\o}ns and Bar-Eli \cite{Brons1994}, their \textit{Inflection Line Method} only applicable to two-dimensional \textit{slow fast dynamical systems}. In 1998, Rossetto \textit{et al.} \cite{Rossetto1998} proposed the \textit{Tangent Linear System Approximation} and then, the \textit{Flow Curvature Method} was developed by Ginoux \textit{et al.} \cite{GiRo1, GiRo2} and Ginoux \cite{Gin}. We have proved in this work that the \textit{Flow Curvature Method} encompasses the \textit{Intrinsic Low-Dimensional Manifold Method}, the \textit{Tangent Linear System Approximation Method} and the \textit{Inflection Line Method} which all belong to the second category we have called \textit{Curvature-Based Methods}. At last, we have also established the identity between the \textit{slow invariant manifold} equation given by the \textit{Geometric Singular Perturbation Method} and by the \textit{Flow Curvature Method} up to suitable order in $\varepsilon$. Then, one of the main criticism made against the \textit{Flow curvature Method}, i.e., the existence and significance of ``ghost parts'' in the \textit{slow invariant manifold} has been completely clarified and moreover, it has been proved that such ``ghost parts'' do also exist for the \textit{Geometric Singular Perturbation Theory}. The existence of such ``ghost parts'' is due to the loss of normal hyperbolicity near the \textit{fold} or, near non-hyperbolic regions.\\

According to Heiter \textit{et al.} \cite{Heiter}, the \textit{Flow Curvature Method} ``uses extrinsic curvature of curves in hyperplanes (codimension-1 manifolds) in order to derive a determinant criterion for computing slow manifold points.'' Last year, Poppe \textit{et al.} \cite{Poppe} generalized the \textit{Flow Curvature Method} based upon \textit{intrinsic curvature} and reformulated it in the framework of \textit{Riemannian Geometry}. With these authors, we ``share the opinion that the field of differential geometry is an appropriate frame to gain further insight in order to adequately define SIMs as slow attracting phase space structures exploited for model reduction purposes.''

\section*{Acknowledgments}
Author would like to thank Pr. Alain Goriely who convinced him of the necessity to write this paper and his wife Dr. Roomila Naeck who contributed to state the proof presented in Sec. 4.

%\end{multicols}
\end{document}